\documentclass[12pt]{amsart}
\usepackage[utf8x]{inputenc}

\usepackage{a4wide}
\usepackage{amssymb}
\usepackage{amsmath}
\usepackage{amsthm}
\usepackage[mathscr]{euscript}
\usepackage[normalem]{ulem}
\usepackage{hyperref}
\usepackage{appendix}
\usepackage{tikz}
\usetikzlibrary{arrows,angles}
\usetikzlibrary{matrix}
\usetikzlibrary{decorations}
\usetikzlibrary{arrows,calc,shapes,decorations.pathreplacing}
\usepackage{tikz-cd}
\usepackage{todonotes}

\usepackage{verbatim}
\usepackage{color}
\usepackage{mathrsfs}

\usepackage[margin=3 cm,heightrounded=true,centering]{geometry}
\usepackage{graphicx}
\usepackage{wrapfig}
\usepackage{cancel}

\numberwithin{equation}{section}

\newtheorem{theorem}{Theorem}[section]

\newtheorem{lemma}[theorem]{Lemma}
\newtheorem{remark}[theorem]{Remark}

\newtheorem{definition}[theorem]{Definition}

\newtheorem{conjecture}[theorem]{Conjecture}
\newtheorem{claim}[theorem]{Claim}

\usepackage[xcolor]{changebar}
\cbcolor{blue}

\DeclareMathOperator{\ric}{Ric}

\DeclareMathOperator{\hess}{Hess}

\DeclareMathOperator{\tr}{tr}
\DeclareMathOperator{\tf}{tf}

\DeclareMathOperator{\divergence}{div}

\DeclareMathOperator{\Span}{span}
\DeclareMathOperator{\graph}{graph}


\title[Penrose theorem and horizon topology in weighted spacetimes]{The Penrose singularity theorem, MOTS stability, and horizon topology in weighted spacetimes}

\author{Eric Ling}
\address[Eric Ling]{Copenhagen Centre for Geometry and Topology (GeoTop), Department of Mathematical Sciences, University of Copenhagen, Copenhagen, Denmark DK-2100}
\email{el@math.ku.dk}
\author{Argam Ohanyan}
\address[Argam Ohanyan]{Department of Mathematics, University of Toronto, Toronto, Ontario, Canada M5S 2E4}
\email{argam.ohanyan@utoronto.ca}
\author{Eric Woolgar}
\address[Eric Woolgar]{Department of Mathematical and Statistical Sciences and Theoretical Physics Institute, University of Alberta, Edmonton, Alberta, Canada T6G 2N8}
\email{ewoolgar@ualberta.ca}

\begin{document}

\date{\today}

\begin{abstract}
We consider versions of the Penrose singularity theorem and the Hawking horizon topology theorem in weighted spacetimes that contain weighted versions of trapped surfaces, for arbitrary spacetime dimension and synthetic dimension. We find that suitable generalizations of the unweighted theorems hold under a weighted null energy condition. Our results also provide further evidence in favour of a weighted scalar curvature that differs from the trace of the weighted Ricci curvature. When the synthetic dimension is a positive integer, these weighted curvatures have a natural interpretation in terms of warped product metrics.
\end{abstract}

\maketitle

\section{Introduction}\label{section: intro}
\setcounter{equation}{0}

\noindent The major singularity theorems of general relativity due to Penrose and Hawking, described in Hawking--Ellis \cite[Chapter 8]{HE} and proved more than 50 years ago, give conditions on the geometry of spacetime that imply some form of causal geodesic incompleteness. This does not provide much information about the nature of spacetime singularities. It leaves open the question of whether such a singularity is mild enough to permit spacetime to be extended past it with perhaps only a slightly loss of smoothness of the differentiabile structure.

Recent work has provided greater insight into these questions. It is now known that some spacetimes cannot be extended, at least not without significant loss of smoothness \cite{Sbierski2018}. A fascinating development has been the work of Cavalletti--Mondino \cite{CM}, which shows that the Hawking cosmological singularity theorem holds in the very weak setting of a Lorentzian length space that need not have any differentiable structure at all. This theorem uses a ``synthetic'' version of the strong energy condition, derived from the theory of optimal transport. The techniques used to prove the theorem draw heavily on measure theory. Indeed, when the structure provided by the Lorentzian metric and its associated volume form cannot be assumed, more basic structures are still necessary, including assumed causal relations and reference measures.

To make further progress, a synthetic notion of the null energy condition is desirable. Versions of those have been put forward by McCann \cite{McCann}, Ketterer \cite{Ketterer} and Cavalletti--Manini--Mondino \cite{CMM1, CMM2}. Ketterer was able to use the proposed synthetic null energy condition to prove Penrose's black hole singularity theorem and the Hawking monotonicity formula in the smooth setting, and to obtain results for a weighted family of synthetic conditions known as synthetic Bakry--\'Emery null energy conditions. Meanwhile, the approach of Cavalletti--Manini--Mondino extends to the synthetic setting, where those authors are able to obtain a version of the Penrose singularity theorem for $C^0$-spacetime metrics on smooth manifolds, assuming a Bakry--\'Emery type null energy condition.

In the setting of smooth manifolds, the Bakry--\'Emery Ricci curvature can be described as a variant of the familiar Ricci curvature that incorporates the Hessian of a \emph{weight function} $f$. Precisely, the smooth $m$-Bakry--\'Emery Ricci curvature is
\begin{equation}
\label{eq1.1}
\ric^{f,m}:=\ric +\hess f -\frac{1}{m}df\otimes df,
\end{equation}
for some $f\in C^{\infty}(M)$ and some $m\in {\mathbb R}\cup \{\infty \}$, where the $m=0$ case is often interpreted as being the case of $f=const$ (so $\ric^{f,0}:=\ric$) and $m=\infty$ to mean $\ric^{f,\infty}:=\ric +\hess f$ (note that $\ric^{f,\infty}= cg$ is the gradient Ricci soliton equation). Since the $m=0$ case is standard, we will omit it from our considerations in this paper.

Recall that the Penrose singularity theorem uses the notion of a \emph{closed trapped surface}, i.e., a closed surface in which both null expansions are negative. In our setting, we will use the notion of a \emph{closed $f$-trapped surface} in which the null expansions are modified by $f$; see Section \ref{section2}. A Penrose-type singularity theorem, assuming nonnegativity of the $m$-Bakry--\'Emergy Ricci curvature in the $m=\infty$ case, was proved by Case \cite{Case}. A version of Galloway's null splitting theorem \cite{Galloway1} (which can be interpreted as the rigidity of Penrose's singularity theorem), assuming only a bound on the $m$-Bakry--\'Emery Ricci curvature along null directions (for any $m$), was proved in \cite{WW}.
But in the classical smooth setting, a version of the Penrose singularity theorem assuming only a bound on the $m$-Bakry--\'Emery Ricci curvature along null directions (for any $m$) seems not to have appeared.

In this paper we revisit the analysis in the smooth case (also called the classical case) of this condition. First, we give a proof of the $m$-Bakry--\'Emergy version of the Galloway null splitting theorem of \cite{WW} using a direct argument by utilizing the ellipticity of the Bakry--Émery null mean curvature operator (more closely following Galloway's argument, rather than the conformal techniques utilized in \cite{WW}).
Then we establish the Penrose singularity theorem in the classical case for (almost) all values of $m$, i.e., we prove the following.

\begin{theorem}[Bakry--\'Emery Penrose singularity theorem]\label{theorem1.1}
Let $(M,g)$ be a globally hyperbolic spacetime of dimension $\dim M = n \geq 3$ with noncompact Cauchy surfaces, $f \in C^2(M)$ and $m \in (-\infty, 2 - n] \cup (0,\infty]$. Suppose that $(M,g)$ satisfies the Bakry--\'Emery null energy condition $\ric^m_f(\ell,\ell) \geq 0$ for all null vectors $\ell \in TM$ and contains a compact spacelike codimension-$2$ surface $S$ that is future $f$-trapped; i.e., $\theta^\pm_f|_S < 0$ (where the $\pm$ sign denotes the two future-null directions orthogonal to $S$). Then:
\begin{itemize}
\item [i)] If $m>0$  then $(M,g)$ is future null geodesically incomplete.
\item [ii)] If $m\le 2-n$ or $m=\infty$ then $(M,g)$ is future null geodesically $f$-incomplete (see Definition \ref{definition2.3} or \cite[Definition 1.3]{WW}), and is geodesically incomplete if $f$ is bounded above.
\end{itemize}
\end{theorem}
We are also able to prove a one-sided version; see Theorem \ref{theorem3.5}.

The $m=2-n$ case of Theorem \ref{theorem1.1} is in fact just the usual Penrose singularity theorem. To see this, observe from formula \eqref{eqB.7} of Appendix \ref{appendixB} that the conformal transformation ${\tilde g}=e^{-2f/(n-2)}$, together with the replacement of $m=0$ by $m=2-n$, preserves the null energy condition in the sense that
\begin{equation}
\label{eq1.2}
\ric_{\tilde g}(\ell,\ell)= \ric_g^{f,2-n}(\ell,\ell)
\end{equation}
for any null vector $\ell$. Now consider the null mean curvature $\theta$ defined by a congruence of null geodesics with tangent field $\ell$ issuing from a closed trapped surface $\Sigma$. Under this conformal transformation, we have $\theta \mapsto\theta-\nabla_{\ell} f =\theta_f$. Now say that $f$ is bounded above. Then null geodesic incompleteness implies null geodesic $f$-incompleteness; see Remark \ref{remark2.4}. The converse is true if $f$ is bounded below. In this sense, Theorem \ref{theorem1.1} with $m=2-n$ can be regarded as a conformally invariant statement of the usual Penrose theorem.

If the null expansion of a closed surface $\Sigma$ defined by the future-outbound null geodesics is identically zero, then we say the surface is \emph{marginally trapped}. It is \emph{outermost} if there are no trapped surfaces or marginally trapped surfaces outside of and homologous to $\Sigma$. Marginally trapped surfaces admit a well-known notion of stability, i.e., it is \emph{stable} if the first eigenvalue of its stability operator is nonnegative. It's not hard to see that stability implies outermost. Similarly, there are analogous notions of outermost and stable for \emph{marginally $f$-trapped} surfaces. The second variation formula for marginally $f$-trapped surfaces was studied in \cite{RW} for the $m=\infty$ case. There it was found that the constraints on the topology of marginally outer $f$-trapped surfaces that follow in the $m=0$ case do not follow in the $m=\infty$ case unless further restrictive assumptions are made. We revisit this question here, and present the second variation formula for marginally $f$-trapped surfaces for arbitrary $m$. We are able to improve the result of \cite{RW}, both by generalizing in $m$ and by using a different version of the weighted dominant energy condition which we call the \emph{$m$-Bakry--\'Emery dominant energy condition}; see equation \eqref{eq4.7} and below it for a precise definition. We also correct \cite{RW} by including a term in the calculation that was overlooked in \cite{RW}.

\begin{theorem}\label{theorem1.2}
Let $\Sigma$ be an outermost marginally $f$-trapped hypersurface, $\dim\Sigma \geq 2$, in an initial data set $(V,g\vert_V,K)$ within a weighted spacetime $(M,g,f)$ such that
\begin{itemize}
\item [a)] the $m$-Bakry--\'Emery dominant energy condition holds for some $m \in (0,\infty]$,
\item [b)] $\nabla_{\ell}f\vert_{\Sigma}=0$.
\end{itemize}
Then there is a Riemannian metric $h$ on $\Sigma$ and function $\tilde{f} \in C^\infty(\Sigma)$ such that the  $m$-Bakry--\'Emery scalar curvature
\begin{equation}
\label{eq1.3}
R^{\tilde{f},m}(h):=R_h+2\Delta_h \tilde{f} -\frac{(m+1)}{m}|d\tilde{f}|^2_h ,
\end{equation}
is positive unless the $m$-Bakry--\'Emery scalar curvature of the induced metric $g|_\Sigma$ and induced weight function $f|_\Sigma$ is identically zero.
\end{theorem}

This theorem generalizes a well-known theorem of Galloway and Schoen \cite{GS} to the Bakry-{\'E}mery setting. As the proof shows, Theorem \ref{theorem1.2} remains true
if in place of the outermost property we assume only stability
for the marginally outer $f$-trapped surface $\Sigma$. The papers \cite{Galloway3} and \cite{Galloway4} explore the rigidity cases of the result in \cite{GS}. We were not able to obtain these rigidity results in the Bakry-{\'E}mery setting without assuming unreasonably strong hypotheses on the weight function $f$.

Many remarks are in order.
\begin{remark}\label{remark1.3}
The proof we give for Theorem \ref{theorem1.2} also holds for $m \in (-\infty, -\dim\Sigma] \cup (1-\dim\Sigma,0)$.
\end{remark}

\begin{remark}\label{remark1.4}
$R^{f,m}>0$ is a nonvacuous constraint.
\end{remark}
For $m$ positive at least, there are manifolds that do not admit metrics and weight functions $(g,f)$ that have positive $R^{f,m}$; see Section \ref{section5} for examples. In the $m=0$ case (i.e., if $f=const$), manifolds with identically zero scalar curvature that do not admit positive scalar curvature are Ricci-flat. We were not able to find an analogous result in the $m$-Bakry--\'Emery setting.

\begin{remark}\label{remark1.5}
Though assumption (b) may be viewed as undesirable in general, we will see in Section \ref{section5} that assumption (a) is much weaker than the usual dominant energy condition, so it should not be surprising that an additional assumption is sometimes required.
\end{remark}

\begin{remark}\label{remark1.6}
When $\dim\Sigma=2$ (so the spacetime dimension is $4$) and $f$ is constant, the proof of Theorem \ref{theorem1.2} is known to simplify considerably (using the Gauss-Bonnet theorem), but we were unable to find such a simplification here when $f$ is nonconstant.
\end{remark}

We can apply Theorem \ref{theorem1.2} to obtain analogues of results that are known to hold when $\dim\Sigma = 3$ and $m=0$ (i.e., when $f=const$).
First it will be useful to recall that an \emph{aspherical} manifold is one whose higher homotopy groups $\pi_k(M)$, $k\ge 2$, are all trivial. It will also be useful to recall the following conjectures.

\begin{conjecture}[J Rosenberg \cite{Rosenberg}]\label{conjecture1.7}
A closed manifold $\Sigma$ of dimension $\ge 3$ and not equal to $4$ admits a positive scalar curvature metric if and only if $\Sigma\times{\mathbb S}^1$ does.
\end{conjecture}
Rosenberg himself found a counter-example in dimension $4$ but the conjecture is still open in all higher dimensions. The theorem below relies only on the conjecture holding in the special case where $\Sigma\times{\mathbb S}^1$ admits a \emph{warped product metric} of positive scalar curvature.

\begin{conjecture}[see \cite{Chodosh} Conjecture 7.28]\label{conjecture1.8}
No closed aspherical manifold admits a positive scalar curvature metric.
\end{conjecture}

This conjecture has been proved in dimensions $4$ and $5$ by Chodosh and Li \cite{CL}. We can now state the following topological result whose proves relies on \cite{CL}.

\begin{theorem}\label{theorem1.9}
Let $\Sigma$ be a closed orientable prime manifold. If $\dim\Sigma = 3$, then, for $0 < m \leq 2$, there is a Riemannian metric $g$ and function $f$ on $\Sigma$ with $R^{f,m}(g) > 0$ if and only if  $\Sigma$  admits a positive scalar curvature metric. If $\dim\Sigma = 2$, then the same statement holds for $0 < m \leq 3$.
\end{theorem}

\begin{remark}\label{remark1.10}\:
\begin{itemize}
\item[1.] Using the results of \cite{RX} on the ${\mathbb S}^1$-stability conjecture, the assumption that $\Sigma$ is prime in Theorem \ref{theorem1.9} can be dropped if the restriction $0 < m \leq 1$ is imposed.
\item[2.] If Conjecture \ref{conjecture1.8} holds, then Theorem \ref{theorem1.9} holds for all $m \in (0,\infty)$.
\item[3.] The orientability assumption on $\Sigma$ in Theorem \ref{theorem1.9} can be removed provided $\Sigma$ contains no two-sided $\mathbb{RP}^2$.
\end{itemize}
\end{remark}

To prove the above assertions, we exploit two facts. The first is that for $k$ a positive integer, the $k$-Bakry--\'Emery scalar curvature $R^{f,k}$ can be related to ordinary (i.e., $k=0$) scalar curvature of higher dimensional manifolds endowed with warped product metrics. The second is that when $m$ is not an integer, we have that $R^{f,k}=R^{f,m}+\left ( \frac{1}{m}-\frac{1}{k}\right ) |df|^2\ge R^{f,m}$ for any integer $k>m>0$. We are also able to give an interpretation of  the \emph{$m$-Bakry--\'Emery dominant energy condition} in terms of warped product spacetimes (in physics terms, a simple class of Kaluza-Klein models). Our results apply in particular to the important special case of compact cross-sections of a static Killing horizon. We can compare this approach to the approach of \cite{CGS}, which applies to product metrics without warping and uses the method of topological censorship to obtain constraints on the topology of black hole event horizons (whose compact cross-sections are marginally outer trapped surfaces, if the spacetime is stationary).

This paper is organized as follows. Section \ref{section2} contains basic background describing modifications to the geometry of null geodesics and marginally trapped surfaces in weighted spacetimes. Section \ref{section3.1} contains the aforementioned novel proof of the weighted version of Galloway's null splitting theorem \cite{Galloway1}. Section \ref{section3.2} contains the proof of Theorem \ref{theorem1.1}, as well as a one-sided version of the theorem. In Section \ref{section4.1}, we turn attention to the stability operator for marginally $f$-trapped surfaces and formulate a weighted version of the dominant energy condition. Under the assumption that the latter holds, in Section \ref{section4.2} we make a further assumption appropriate to Killing horizons and are then able to prove Theorem \ref{theorem1.2}. Finally, in Section \ref{section5}, we show how our results apply to warped product spacetimes and discuss the significance and nontriviality of nonnegative and positive Bakry--\'Emery scalar curvature $R^{f,m}$. The proof of Theorem \ref{theorem1.9} is given in Section \ref{section5}. Appendix \ref{appendixA} contains a geometric flow argument showing that nonnegative Bakry-\'Emery scalar curvature that does not vanish everywhere on a closed manifold can be transformed to positive Bakry-\'Emery scalar curvature. Appendix \ref{appendixB} contains formulas that show how Bakry-\'Emery Ricci curvature behaves under conformal transformations coupled to scale transformations of the weight function $f$. Appendix \ref{appendixC} contains a computation used in Section \ref{section4}.

\subsection*{Acknowledgments}
The research of EL was supported by Carlsberg Foundation CF21-0680 and Danmarks Grundforskningsfond CPH-GEOTOP-DNRF151. The research of EW was supported by NSERC grant RGPIN--2022--03440. All three authors are grateful to the University of Vienna and the Emerging Fields Project ``A New Geometry for Einstein's Theory of Relativity and Beyond'' for hospitality in July 2025.

This research was funded in part by the Austrian Science Fund (FWF) [Grants DOI \href{https://doi.org/10.55776/EFP6}{10.55776/EFP6} and \href{https://doi.org/10.55776/J4913}{10.55776/J4913}]. For open access purposes, the authors have applied a CC BY public copyright licence to any author accepted manuscript version arising from this submission.

\section{Weighted null mean curvature}\label{section2}
\setcounter{equation}{0}

\subsection{Definition of weighted null mean curvature}
We begin by recalling the null second fundamental form. Let $S\hookrightarrow M$ be a null hypersurface and let $p\in S$. One can define a basis $\left \{ k,\ell, e_i; i=1,\dots, n-2 \right \}$ for $T_pM$ such that $g(k,k)=g(\ell,\ell)=0$ so $k$ and $\ell$ are future-pointing null; $\ell$ is tangent to $S$; the vectors $e_i$ are tangent to $S$, spacelike, normalized to unit length, and orthogonal to each other and to $k$ and $\ell$; and $g(\ell,k)=-1$. Then $\left \{ \ell,e_i; i=1,\dots,n-2\right \}$ is a basis for $T_pS\subset T_pM$. For $X\in T_pS$ let $\overline{X}\in T_pS/\ell$ denote the equivalence class of all $X'\in T_pS$ such that $X'\sim X$ iff $X'-X\in \Span \{ \ell \}$.

\begin{definition}\label{definition2.1}
For any $X\in T_pS$, the \emph{tangential null Weingarten map} $B_{\ell}:T_pS/\ell\to T_pS/\ell$ is given by
\begin{equation}\label{eq2.1}
B_{\ell}(X):=\overline{\nabla_X \ell}.
\end{equation}
\end{definition}
There is also a \emph{transverse null Weingarten map} where $\ell$ is replaced by $k$. Null Weingarten maps are discussed, for example, by Galloway \cite[Section II.1]{Galloway1}.

The trace of a tangential null Weingarten map yields the corresponding null tangential mean curvature.
\begin{definition}{\cite[Equation II.1]{Galloway1}}\label{definition2.2}
The \emph{tangential null mean curvature}, or \emph{expansion scalar}, of $S$ is
\begin{equation}
\label{eq2.2}
\theta:=\sum\limits_{i=1}^{n-2} h \left ( B_{\ell}({\overline e}_i),{\overline e}_i\right ) = \sum\limits_{i=1}^{n-2}  g(\nabla_{e_i}\ell,e_i).
\end{equation}
where $h$ is the positive-definite inner product on $T_pS/\ell$ induced from the inner product $g(\cdot,\cdot)$ on $T_pM$. The \emph{Bakry--\'Emery tangential null mean curvature}, or \emph{Bakry--\'Emery expansion scalar}, of $S$ is
\begin{equation}\label{eq2.3}
\theta_f:=\theta-\nabla_\ell f.
\end{equation}
\end{definition}
Later in the sequel the \emph{transverse null mean curvature} will be denoted by $\kappa$ and the corresponding \emph{Bakry--\'Emery transverse null mean curvature} will be denoted by $\kappa_f$.

\subsection{The Raychaudhuri equation}
The definition of $\theta_f$ is motivated by the following calculation. Consider a smooth spacelike hypersurface $P$ intersecting the null hypersurface $S$ and let $\Sigma=S\cap P$. Then $\Sigma$ has null second fundamental form $B$ in $S$, governed by the familiar Raychaudhuri equation. Let $\tf B:=B-\frac{1}{(n-2)}\theta h$ denote the tracefree part of $B$ and let $\theta:=\tr B$ be the null mean curvature. We may choose the basis null vector field $\ell$ such that it is tangent to an affinely parametrized null geodesic congruence generating $S$ (and parametrized by $t$). Then the Raychaudhuri equation yields
\begin{equation}\label{eq2.4}
\frac{d\theta}{dt}= -\ric(\ell,\ell)-|B|^2.
\end{equation}
It can be re-expressed in terms of $\theta_f$ and the Bakry--\'Emery Ricci curvature
\begin{equation}\label{eq2.5}
\ric_f^m:=\ric+\hess f - \frac{1}{m}df\otimes df.
\end{equation}
Using that the null generators $\gamma$ of $S$ are geodesics with tangent vector field $\ell =\frac{d}{dt}$ for $t$ an affine parameter along each generator, then
\begin{equation}\label{eq2.6}
\begin{split}
\frac{d\theta_f}{dt}=&\, -\ric(\ell,\ell)-\hess f (\ell,\ell) -|B|^2\\
=&\, -\ric_f^m (\ell,\ell) -\frac{1}{m}\left ((f\circ\gamma)'\right )^2 -|\tf B|^2-\frac{1}{n-2}\theta^2\\
=&\, -\ric_f^m (\ell,\ell) -\left ( \frac{1}{m}-\frac{1}{n-2}\right )\left ((f\circ\gamma)'\right )^2 -|\tf B|^2-\frac{\theta_f^2}{n-2}\\
&\, -\frac{2\theta(f\circ\gamma)'}{n-2}\\
=&\, -\ric_f^m (\ell,\ell) -\left ( \frac{1}{m}+\frac{1}{n-2}\right )\left ((f\circ\gamma)'\right )^2 -|\tf B|^2-\frac{\theta_f^2}{n-2}\\
&\, -\frac{2\theta_f(f\circ\gamma)'}{n-2}.
\end{split}
\end{equation}

If $m>0$ (and $n > 2$, of course), we use the identity
\begin{equation}
\label{eq2.7}
\frac{\theta^2}{n-2}+\frac{\left ( (f\circ \gamma)'\right )^2}{m}\ge \frac{\left ( \theta- (f\circ \gamma)'\right )^2}{n+m-2} \equiv \frac{\theta_f^2}{n+m-2}.
\end{equation}
Then \eqref{eq2.6} quickly simplifies to yield
\begin{equation}
\label{eq2.8}
\begin{split}
\frac{d\theta_f}{dt} =&\, -\ric_f^m (\ell,\ell) -\frac{1}{m}\left ((f\circ\gamma)'\right )^2 -|\tf B|^2-\frac{1}{n-2}\theta^2\\
\le &\, -\frac{\theta_f^2}{n+m-2}.
\end{split}
\end{equation}

If we cannot assume that $m>0$, \eqref{eq2.6} implies that
\begin{equation}
\label{eq2.9}
\begin{split}
\frac{d}{dt}\left ( e^{\frac{2f\circ\gamma}{n-2}}\theta_f\right ) =&\,  -e^{\frac{2f\circ\gamma}{n-2}}\left [ \ric_f^m (\ell,\ell) +\frac{(n+m-2)}{m(n-2)}\left ((f\circ\gamma)'\right )^2\right .\\
&\, \left .+|\tf B|^2+\frac{\theta_f^2}{n-2}\right ].
\end{split}
\end{equation}
Define $x_f:=e^{\frac{2f\circ\gamma}{n-2}}\theta_f$. Then we obtain
\begin{equation}
\label{eq2.10}
\frac{dx_f}{dt} \le -e^{\frac{2f\circ\gamma}{n-2}} \ric_f^m (\ell,\ell) -e^{\frac{-2f\circ\gamma}{n-2}}\frac{x_f^2}{n-2}
\end{equation}
whenever either $m>0$ or $m\le 2-n$, with equality if and only if $\tf B=0$ and either $(f\circ\gamma)'=0$ or $m=2-n$. When $\ric_f^m(\ell,\ell)\ge 0$, then
\begin{equation}
\label{eq2.11}
\frac{dx_f}{dt} \le -\frac{x_f^2e^{\frac{-2f\circ\gamma}{n-2}}}{n-2}.
\end{equation}
Now define a new parameter $\tau$ along the null geodesic generators $\gamma(t)$ of $S$ by
\begin{equation}\label{eq2.12}
\tau:=\int\limits_0^t e^{-2f\circ \gamma (t')/(n-2)}dt'.
\end{equation}
Then \eqref{eq2.11} becomes
\begin{equation}\label{eq2.13}
\frac{dx_f}{d\tau} \le -\frac{x_f^2}{n-2}, \quad x_f:=e^{\frac{2f\circ\gamma}{n-2}}\theta_f,
\end{equation}
whenever $m\le 2-n$ (or $m\in (0,\infty]$).

When $f\le C$, inequality \eqref{eq2.13} can be used to control $\theta_f$ for $m\in (-\infty,n-2]$ or $m=\infty$ (i.e., when coefficients of $\frac{1}{m}$ are replaced by $0$). Inequality \eqref{eq2.13} also applies when $m\in (0,\infty)$, but then \eqref{eq2.8} gives better control and does so without needing boundedness of $f$.

\begin{definition}\label{definition2.3}
We say that an inextendible geodesic is \emph{future $f$-complete} if there is a $\tau_0\in {\mathbb R}$ (or possibly $\tau_0=-\infty$) such that the geodesic is defined for all $\tau\ge \tau_0$, where $\tau$ is defined by \eqref{eq2.12}. If the domain of $\tau$ is bounded above, we will say that the inextendible geodesic is \emph{future $f$-incomplete}. Past $f$-complete and past $f$-incomplete are defined dually. A geodesic that is both future and past $f$-complete is called \emph{$f$-complete}. An $f$-complete achronal null curve is a \emph{complete null $f$-line}.
\end{definition}

\begin{remark}\label{remark2.4}
If $f$ is bounded above on $M$ then geodesic completeness implies $f$-completeness and $f$-incompleteness implies geodesic incompleteness.
\end{remark}

\subsection{Ellipticity of the weighted null mean curvature}\label{subsection: ellipticity of theta_f}
A brief calculation of Galloway \cite{Galloway1} shows that $\theta$ is elliptic in the following sense. Following \cite[pp 547--548]{Galloway1}, we now take $P$ to be a timelike hypersurface intersecting the null hypersurface $S$, and again let $\Sigma$ be the intersection surface $\Sigma:=P\cap S$, and let $p \in \Sigma$. We also choose a smooth spacelike hypersurface $V\hookrightarrow P$ (i.e., $V$ is a codimension-two spacelike surface in spacetime) intersecting $\Sigma$ such that $p\in V\cap \Sigma$. Now $\Sigma$ can be expressed, locally near $p$ at least, as a graph over $V$. In Gaussian coordinates $(t,x^1,\dots,x^{n-2})$ for $P$ based at $V$ (so that $V$ is the $t=0$ locus and $\frac{\partial}{\partial t}$ is future timelike), we write
\begin{equation}\label{eq2.14}
\Sigma=\graph u =\left \{ (u(x),x)\in P; x:=(x^1,\dots,x^{n-2})\in V\right \}.
\end{equation}
Using this construction, Galloway is able to show first that the mean curvature of $\Sigma$ in $P$ is elliptic and then that the null mean curvature $\theta$ of $\Sigma$ in $S$ is elliptic on functions $u\in C^{\infty}(V)$. Specifically,
\begin{equation}\label{eq2.15}
\theta=\theta(u) = \sum\limits_{i,j=1}^{n-2}a^{ij}(x,u,\partial u)\partial_{ij}u +b(x,u,\partial u),
\end{equation}
where the matrix $\left [ a^{ij}\right ]$ is positive-definite.

It remains to verify that $\nabla_{\ell}f$ does not depend on derivatives of $u$ beyond first derivatives, which amounts to verifying that $\ell$ does not depend on derivatives of $u$ beyond first derivatives. Since $\Sigma$ is a codimension-two surface in spacetime, there are exactly two future-null directions spanned by vectors orthogonal to $\Sigma$. Extend the coordinate system above to a neighbourhood of $V$ (near $p$) by, say, a Gaussian normal coordinate $x^{n-1}$ based at $P$ (so $P$ is the locus $x^{n-1}=0$). At any $q\in \Sigma$, a spacetime vector $v$ is orthogonal to $\Sigma$ if and only if
\begin{equation}\label{eq2.16}
\left \langle v,dS\right \rangle=0,\quad dS=dt-\sum\limits_{i=1}^{n-2}\frac{\partial u}{\partial x^i}dx^i.
\end{equation}

This condition depends only on $u$ and its first derivatives. Then the condition that $v$ must be null, depending only on the spacetime metric, cannot depend on higher derivatives of $u$ either. These two conditions determine $\ell$ up to a discrete choice (which cannot depend continuously on u or its derivatives) and an overall scale $v=\lambda \ell$. But since integral curves of $v$ meet $\Sigma$ only once, say at $q$, the value of this scale at any point of $\Sigma$ is simply initial data for the scale factor along that integral curve. The initial value can depend on the coordinates $(t(q)=u(x(q)),x(q))$, but not on derivatives of $u$. Hence $\ell=\ell(u,\partial u)$ and so $\theta_f:=\theta-\nabla_{\ell}f$ is an elliptic operator on $C^{\infty}(V)$.

\section{Penrose singularity theorems and null rigidity in weighted spacetimes}\label{section3}
\setcounter{equation}{0}

\noindent In this section, we first give a new proof of Galloway's null splitting theorem \cite{Galloway1} in the Bakry--\'Emery setting (established in\cite{WW}) using the ellipticity of $\theta_f$ and more closely following Galloway's original proof rather than using conformal arguments. Then we turn to the Penrose singularity theorem for arbitrary $m$, i.e., the proof of Theorem \ref{theorem1.1}.

\subsection{The null splitting theorem}\label{section3.1}

The following Bakry--\'Emery version of Galloway's null splitting theorem \cite[Theorem IV.1]{Galloway1} has been established by Woolgar--Wylie \cite{WW} by applying a conformal argument, which reduces the claim to the non-weighted version of the result by Galloway.

\begin{theorem}\label{theorem3.1}
Let $(M,g,f)$ be a smooth spacetime satisfying the Bakry--\'Emery null energy condition for some $m \in  (-\infty,2-n] \cup (0,\infty] $.
\begin{itemize}
\item [i)] When $m\in (0,\infty)$, suppose that $(M,g,f)$ is null geodesically complete and suppose that there exists a null line $\eta$ in $M$. Then $\eta$ is contained in a smooth closed achronal totally geodesic null hypersurface $S \subset M$, and $f$ is constant along each null geodesic generator of $S$.
\item [ii)] When $m\in (-\infty,2-n)\cup \{ \infty\}$ assume that that $(M,g,f)$ is null geodesically $f$-complete and suppose that there exists a complete null $f$-line $\eta$ in $M$. Then $\eta$ is contained in a smooth closed achronal totally geodesic null hypersurface $S \subset M$, and $f$ is constant along each null geodesic generator of $S$.
\item [iii)] When $m=2-n$, assume that that $(M,g,f)$ is null geodesically $f$-complete and suppose that there exists a complete null $f$-line $\eta$ in $M$ (i.e., an achronal null geodesic that is complete with respect to $\tau$ as defined in \eqref{eq2.12}. Then $\eta$ is contained in a smooth closed achronal umbilic null hypersurface $S \subset M$.
\end{itemize}
\end{theorem}
By \emph{umbilic} we will mean that the null Weingarten map is the identity times a function $F:S\to {\mathbb R}$.

We will now give a new proof of Theorem \ref{theorem3.1} by making use of the ellipticity of $\theta_f$ observed before. To this end, we require the following Bakry--\'Emery version of Galloway's maximum principle for null hypersurfaces \cite[Theorem II.1]{Galloway1}.

\begin{lemma} \label{lemma: BE null maximum principle} Let $S_1$ and $S_2$ be two null hypersurfaces in a smooth spacetime $(M,g)$. Suppose $S_1$ and $S_2$ meet at $p$, and that $S_2$ lies to the future side of $S_1$ near $p$. Moreover, suppose that their respective Bakry--\'Emery tangential null mean curvatures $\theta_f^1$ and $\theta_f^2$ satisfy
\begin{equation}
\label{eq3.1}
\theta^1_f \leq 0 \leq \theta^2_f.
\end{equation}
Then $S_1$ and $S_2$ coincide near $p$, and this common null hypersurface has vanishing Bakry--\'Emery tangential null mean curvature $\theta_f=\theta_f^1 = \theta_f^2 = 0$.
\end{lemma}

\begin{remark}\label{remark3.3}
The same conclusion holds if $S_1$ is a $C^0$ future null hypersurface and $S_2$ is a $C^0$ past null hypersurfaces and the inequality $\theta_f^1 \leq 0 \leq \theta_f^2$ holds in the sense of support hypersurfaces, if the support hypersurfaces of $S_1$ have null second fundamental forms locally bounded from below. The common hypersurface where $S_1$ and $S_2$ coincide (locally) is smooth.
\end{remark}

\begin{proof}
We observed in Subsection \ref{subsection: ellipticity of theta_f} that, just like the null mean curvature $\theta$, the Bakry--\'Emery null mean curvature $\theta_f$ is a second order quasilinear elliptic operator. The proof of the maximum principle follows from Alexandrov's maximum principle (see \cite[Thm.\ II.2]{Galloway1}).

To prove the Remark, note that the adjustment of this argument to the case of $C^0$-null hypersurfaces is obtained in the same way as \cite[Theorem 3.4]{Galloway1}.
\end{proof}

As a preliminary result before proving the null splitting, we will need the following lemma.

\begin{lemma} \label{lemma3.4}
Suppose that $S$ is an achronal $C^0$
null hypersurface in $(M,g,f)$. Let $m\in (-\infty,2-n] \cup (0,\infty]$ be such that $(M,g,f)$ satisfies the $m$-Bakry--\'Emery null energy condition.
\begin{itemize}
\item [i)] If $m\in (0,\infty)$,  assume that the null generators of $S$ are future complete.
\item [ii)] If $m\in  (-\infty,2-n]$ or $m=\infty$, assume that the null generators of $S$ are future $f$-complete.
\end{itemize}
Then $S$ satisfies $\theta_f \geq 0$ in the sense of support hypersurfaces, with null second fundamental forms locally bounded from below.
\end{lemma}

\begin{proof}
We follow Galloway's proof of \cite[Lemma IV.2]{Galloway1}, mentioning only those modifications needed. As in that proof, we assume global hyperbolicity of $M$ and let $p \in S$ and let $K$ be a tangent to an affinely parameterized
future-inextendible
null generator $\eta:[0,\infty) \to M$, $\eta(0) = p$, of $S$ through $p$.
Let $S_{p,K,r}$ be a past support hypersurface $U\cap\partial J^-(\eta(r))$ of $S$ at $p$, for $U$ a suitable neighbourhood of $p$ and $r \in (0,T)$. Denote by $\theta_{f}(s)$ the Bakry--\'Emery null mean curvature of $S_{p,K,r}$ at $\eta(s)$ with respect to $\eta'(s)$.

When $m\in (-\infty,2-n]\cup \{ \infty\}$, we claim that
\begin{equation}
\label{eq3.2}
\theta_f(0) \geq - \frac{(n-2)e^{-\frac{2f \circ \eta(0)}{(n-2)}}}{\int_0^r e^{-\frac{2f \circ \eta(s)}{(n-2)}} ds}.
\end{equation}
To prove the claim, by \eqref{eq2.11} the null mean curvature of the support surface obeys
\begin{equation}
\label{eq3.3}
\frac{dx_f}{dt} \leq - \frac{x_f^2(t) e^{-\frac{2f \circ \eta(t)}{(n-2)}}}{(n-2)},\quad x_f(t) = e^{\frac{2f \circ \eta(t)}{n-2}} \theta_f
\end{equation}
whenever $m\in (-\infty,2-n]\cup \{ \infty\}$. Without loss of generality, we may take $x_f(0)<0$ for otherwise the claim is trivially true, and then we see from \eqref{eq3.3} that $x_f(t)<0$ for all $t\ge 0$. Let $s:=r-\delta -t$ for $\delta > 0$. Then
\begin{equation}
\label{eq3.4}
\frac{dx_f}{ds} \geq  \frac{x_f^2(r-\delta-s)e^{-\frac{2f \circ \eta(r-\delta-s)}{(n-2)}}}{(n-2)},
\end{equation}
so
\begin{equation}
\label{eq3.5}
\frac{d}{ds} \left ( \frac{1}{x_f}\right ) \le - \frac{e^{-\frac{2f \circ \eta(r-\delta-s)}{(n-2)}}}{(n-2)}.
\end{equation}
We integrate backwards along $\eta$ from $s=0$ (where $t=r-\delta$) to $s=r-\delta$ (where $t=0$). Then
\begin{equation}
\label{eq3.6}
\begin{split}
\frac{1}{x_f(0)}-\frac{1}{x_f(r-\delta)}\le &\, -\frac{1}{(n-2)}\int\limits_0^{r-\delta} e^{-\frac{2f \circ \eta(r-\delta-s)}{(n-2)}}ds = \frac{1}{(n-2)}\int\limits_{r-\delta}^0 e^{-\frac{2f \circ \eta(t)}{(n-2)}}dt\\
=&\, -\frac{1}{(n-2)}\int\limits_0^{r-\delta} e^{-\frac{2f \circ \eta(t)}{(n-2)}}dt.
\end{split}
\end{equation}
Now we use our earlier observation that $x_f(t)<0$. The inequality \eqref{eq3.6} then yields
\begin{equation}
\label{eq3.7}
\frac{1}{x_f(0)}\le -\frac{1}{(n-2)}\int\limits_0^{r-\delta} e^{-\frac{2f \circ \eta(t)}{(n-2)}}dt.
\end{equation}
The claim follows by multiplying both sides of \eqref{eq3.7} by $e^{\frac{2f \circ \eta(0)}{(n-2)}}$, taking reciprocals, and taking the limit $\delta\searrow 0$.

It is then a consequence of future $f$-completeness of $\eta$ that the lower bound approaches $0$ as $r \nearrow T$. Thus, we conclude that $S$ satisfies $\theta_f \geq 0$ in the sense of the support hypersurfaces $S_{p,K,r}$. The fact that the second fundamental forms of $S_{p,K,r}$ are locally bounded from below and that the above arguments extend to the case where $M$ is not globally hyperbolic follow from the same arguments given by Galloway in the proof of \cite[Lem.\ IV.2]{Galloway1}.

Comparing inequalities \eqref{eq2.8} and \eqref{eq2.13}, we see the the manipulations immediately above will yield \eqref{eq3.8} with $x_f$ replaced by $\theta_f$ and $n-2$ replaced by $n+m-2$. The inequality \eqref{eq3.5} becomes simply
\begin{equation}
\label{eq3.8}
\frac{d}{ds} \left ( \frac{1}{x_f}\right ) \le - \frac{1}{(n+m-2)}.
\end{equation}
Integrating as in \eqref{eq3.6}, we get
\begin{equation}
\label{eq3.9}
\frac{1}{x_f(0)}-\frac{1}{x_f(r-\delta)} \le -\frac{r-\delta}{(n+m-2)},
\end{equation}
and in place of \eqref{eq3.2}, for $m>0$ we obtain that
\begin{equation}
\label{eq3.10}
\theta_f(0) \ge - \frac{(n+m-2)}{r}.
\end{equation}
Since $\eta$ is future-complete, we take $r\to\infty$ to obtain that $\theta_f(0) \ge 0$ as claimed.
\end{proof}

\begin{proof}[Proof of Theorem \ref{theorem3.1}]

The argument follows \cite{Galloway1} closely, so we provide only an outline illustrating the differences. Recall \eqref{eq2.8}, which reads as
\begin{equation}
\label{eq3.11}
\begin{split}
\frac{d\theta_f}{dt} =&\, -\ric_f^m (\ell,\ell) -\frac{1}{m}\left ((f\circ\gamma)'\right )^2 -|\tf B|^2-\frac{1}{(n-2)}\theta^2\\
\le &\, -\frac{\theta_f^2}{(n+m-2)}.
\end{split}
\end{equation}
We regard this equation as governing Jacobi fields along a null geodesic $\eta$. Standard arguments applied to inequalities of the form $x_f'\le -kx_f^2$, $k>0$, so that $x$ will diverge at a finite value of its argument unless it is always zero. If such a divergence were to occur, it would occur in $\theta$, since $f\in C^2(M)$. This signals a conjugate pair of points along $\eta$, contradicting the assumption that $\eta$ is a line. Hence $\theta_f=0$ along $\eta$ and, since we assume that $\ric_f^m(\ell,\ell)\ge 0$, every term on the right of the first line of \eqref{eq3.11} must vanish. Note in particular that $(f\circ\eta)'$ must vanish, so $f$ is constant along $\eta$ and moreover the vanishing of $\theta_f$ then implies the vanishing of $\theta$. Then the entire null second fundamental form $B$ vanishes along $\eta$, and $\ric(\ell,\ell)=0$ along $\eta$ as well. Finally, the boundary $S$ of the past of $\eta$ is an achronal null hypersurface generated by null lines, so following \cite{Galloway1} one can apply this analysis to each null generator.

Now consider $m\in (-\infty,2-n)\cup \{ \infty \}$. For this case, we refer to \eqref{eq2.9}, which reads as
\begin{equation}
\label{eq3.12}
\begin{split}
\frac{d}{dt}\left ( e^{\frac{2f\circ\gamma}{n-2}}\theta_f\right ) =&\,  -e^{\frac{2f\circ\gamma}{n-2}}\left [ \ric_f^m (\ell,\ell) +\frac{(n+m-2)}{m(n-2)}\left ((f\circ\gamma)'\right )^2\right .\\
&\, \left .+|\tf B|^2+\frac{\theta_f^2}{(n-2)}\right ].
\end{split}
\end{equation}
Since $x_f=e^{\frac{2f\circ\eta}{(n-2)}}\theta_f=e^{\frac{2f\circ\eta}{(n-2)}}\left ( \theta -(f\circ\eta)'\right )$, we can use the parameter $\tau$ defined by \eqref{eq2.12} to write that $x_f={\tilde \theta} -\frac{d}{d\tau} f\circ \eta$, where ${\tilde \theta}:=\frac{d}{d\tau}\log |J|$, with $|J|$ being the determinant of the matrix of Jacobi fields along $\eta$. Then a divergence in $x_f$ indicates a conjugate pair in the reparametrized geodesic $\eta(\tau)$. This would contradict the assumption that $\eta$ is a complete null $f$-line. Since $\ric_f^m(\ell,\ell)\ge 0$ and $m<2-n\le -1$, every term on the right must therefore vanish. Again $(f\circ\eta)'$ must vanish, so $f$ is constant along $\eta$, $\theta$ vanishes along $\eta$, $B$ vanishes along $\eta$, and $\ric(\ell,\ell)=0$ along $\eta$. Again the argument can be extended to $S$ by applying it to each null line generating the boundary of the past of $\eta$.

Finally, we consider the case of $m=2-n$. Arguing as above, we obtain that $x_f$, $\tf B$, and $\ric_f^m(\ell,\ell)$ vanish along $\eta$ (and along each null generator of the boundary of the past of $\eta$), but now we have only that $\tr B := \theta =\frac{d}{dt}f\circ\eta$.
\end{proof}

\subsection{Penrose singularity theorems}\label{section3.2}

\begin{proof}[Proof of Theorem \ref{theorem1.1}]
Suppose $(M,g)$ is future null geodesically complete (resp.\ future null geodesically $f$-complete). It suffices to show that null geodesics normal to $S$ reach focal points, the rest of the argument is the same as in the non-weighted case (see e.g.\ Hawking--Ellis \cite[pp 263--264]{HE}). Recalling \eqref{eq2.8} (with $\theta_f$ denoting either $\theta_f^+$ or $\theta_f^-$), then
\begin{equation}
\label{eq3.13}
\frac{d \theta_f}{dt} \leq -\frac{\theta_f^2}{n+m - 2}.
\end{equation}
Integrating this inequality from $0$ to $t$ yields
\begin{equation}
\label{eq3.14}
\frac{1}{\theta_f(t)} \geq \frac{t}{n+m - 2} + \frac{1}{\theta_f(0)}.
\end{equation}
Since $\theta_f(0) < 0$, for large enough $t$ (which we are allowed to take by our assumption of null geodesic completeness) the right hand side becomes $0$ at some $t=T>0$, which shows that necessarily $\theta_f(t) \to -\infty$ for some $t\searrow\tau$ where $0<\tau\le T$. Since $\nabla f(\gamma(t))$ is bounded for all $t$ in a bounded interval, where $\gamma$ is the future null normal of $S$ corresponding to $\theta_f$, then necessarily $\theta(\tau)$ diverges to  $-\infty$, which shows that $\eta(\tau)$ is a focal point for $S$.

In the case $m \in (-\infty, 2-n]\cup \{ \infty\}$, the argument proceeds with \eqref{eq2.13} instead. By null $f$-completeness, the parameter $\tau$ defined in \eqref{eq2.12} can be made arbitrarily large if the affine parameter $t$ is chosen arbitrarily large if $f$ is bounded above. As above, we are able to obtain the existence of some value $\tau=T$ such that $x_f(T)$ diverges to $-\infty$. Then $\theta_f(T)$ diverges to $-\infty$, which again implies that $\theta(T)$ diverges to $-\infty$, and we obtain the existence of focal points to $S$ in this case.
\end{proof}

Let us also note the following Bakry--\'Emery version of the one-sided Penrose singularity theorem (cf.\ Andersson--Mars--Simon \cite[Theorem 7.1]{AMS}). The proof follows from the reference given as well as a focusing argument analogous to the one we have given above in the case of the usual Penrose theorem. It would have applications to Bakry--\'Emery versions of the Gannon-Lee singularity theorems \cite{Gannon1, Gannon2, CWLee}

\begin{theorem}[One-sided Bakry--\'Emery Penrose singularity theorem]\label{theorem3.5}
Let $(M,g)$ be a globally hyperbolic spacetime satisfying $\ric^m_f(\ell,\ell) \geq 0$ for all null $\ell \in T_pM$ and all $p\in M$, for some $m \in (-\infty,2-n] \cup (0,\infty]$, $n= \dim M \geq 3$, and for $f \in C^2(M)$. Suppose that $M$ contains a smooth spacelike Cauchy surface $V$ which in turn contains a closed hypersurface $\Sigma$ (i.e., $\Sigma$ is a codimension-$2$ surface in $M$) that separates $V$ into two disconnected parts $V \setminus \Sigma = V^- \cup V^+$. Call $\theta_f^+$ the null $f$-expansion of $\Sigma$ with respect to the future null normal pointing into $V^+$. Suppose that $\theta_f^+|_S < 0$ and that $V^+ \cup \Sigma$ is a connected, non-compact manifold with boundary $\Sigma$.
\begin{itemize}
\item [i)] If $m\in (0,\infty)$ then $(M,g)$ is future null geodesically incomplete.
\item [ii)] If $m\in (-\infty,n-2]\cup \{ \infty \}$ then $(M,g)$ is future null geodesically $f$-incomplete.
\end{itemize}
\end{theorem}

Also the rigidity result corresponding to the above one-sided Bakry--\'Emery Penrose singularity theorem can be established in analogy with the non-Bakry--\'Emery (i.e., $m = 0$) case. See \cite[Proposition 3]{GL} and \cite[Section 7]{EGP} for a proof.

\section{Marginally outer trapped surfaces}\label{section4}
\setcounter{equation}{0}

\subsection{Stability formula} \label{section4.1}
The stability operator for marginally trapped surfaces was written down (for 4-dimensional spacetimes) in \cite{Newman}, following a method used to study minimal surfaces in Riemannian manifolds. The basic technique had already been employed by Hawking (see, e.g., \cite[Proposition 9.3.2]{HE}), and has been followed by other authors since, such as \cite{AMS} and \cite{GS}. An $m=\infty$ Bakry--\'Emery version of the technique appears in \cite{RW}. Here we recast that work and extend it to all positive $m$-values (and also correct an oversight in the argument in \cite{RW}). We point the reader to the clear and general discussion in \cite[Section 7.5]{Lee} (for the non-Bakry--\'Emery case), whose notation we will follow.

We begin with an initial data hypersurface $(V,g_V,K)$ in spacetime. Here $g_V$ and $K$ are, respectively, the first and second fundamental forms induced on the hypersurface $V\hookrightarrow M$ by the spacetime metric $g$ on $M$. In $V$, we define a closed two-sided hypersurface $\Sigma$ (hence $\Sigma$ is codimension-$2$ in $M$) with first and second fundamental forms $h$ and $A$ induced by the inclusion $\Sigma\hookrightarrow V$. We denote the mean curvature of $\Sigma\hookrightarrow V$ as $H:=\tr_h A=:\tr_{\Sigma}A$. We will also use the notation $\tr_{\Sigma}K:=h^{ij}K_{ij}$ to denote the $h$-partial trace of $K$. Let $w$ denote the unit future directed normal field to $V$ and $\nu$ the `outward' unit normal field to $\Sigma$ and tangent to $V$. Define $\ell = w + \nu$ and $k = w - \nu$, which are \emph{outbound} and \emph{inbound} null vector fields on $\Sigma$, respectively.  The first fundamental form, or induced metric, on $\Sigma$ can be written as
\begin{equation}
\label{eq4.1}
h_{\mu\nu}=g_{\mu\nu}+\frac12 \left ( k_{\mu}\ell_{\nu}+\ell_{\mu}k_{\nu} \right ).
\end{equation}
The null second fundamental form defined by $\ell$ is
\begin{equation}
\label{eq4.2}
\theta:= \divergence_{\Sigma} \ell = h^{ij}\nabla_i \ell_j,\quad \ell_j:=g_{ij}\ell^i.
\end{equation}
We also define
\begin{equation}
\label{eq4.3}
\kappa:= \divergence_{\Sigma} k = h^{ij}\nabla_i k_j,\quad k_j:=g_{ji}k^i.
\end{equation}

Now consider the variation vector field $\varphi\nu$ for $\varphi\in C^{\infty}(\Sigma)$. This generates a variation $\Sigma_t$ of $\Sigma$ within $V$ via $x \mapsto \exp_x\left ( t\varphi(x)\nu(x) \right )$, where $\exp$ denotes the exponential map of the metric on $V$. To simplify notation, we suppress the inclusion maps $\Sigma\hookrightarrow V$, $V\hookrightarrow M$, $\Sigma \hookrightarrow M$, and the induced pullbacks. We also ignore a possible component of the variation tangent to $\Sigma$. We extend $\ell$ in a neighborhood of $\Sigma$ within $V$ by defining $\ell = w + \nu$, where $\nu$ is determined by the variation. Then the linearization of $\theta$ (its derivative along any extension of the vector field $\varphi\nu$, evaluated at $\Sigma$) is given by \cite[Proposition 7.32]{Lee} as
\begin{equation}
\label{eq4.4}
\begin{split}
D\theta=&\, -\Delta_{\Sigma} \varphi +2 W_\Sigma\cdot \nabla_\Sigma \varphi +\left [ \divergence_{\Sigma} W_{\Sigma} -\left \vert W_{\Sigma}\right \vert^2 +\frac12 R_{\Sigma} \right .\\
&\, \left . -G(w,\ell) -\frac12 \left \vert K_{\Sigma} +A\right \vert^2 +\frac12 \theta\left ( \kappa +2K(\nu,\nu)\right )\right ]\varphi.
\end{split}
\end{equation}
To lessen the burden of notation, we leave implicit the fact that each term in this expression is evaluated on $\Sigma$. Here $W_\Sigma$ is the tangential vector field on $\Sigma$ that is dual to the 1-form $K(\nu,\cdot)$ along $\Sigma$,
$G$ denotes the Einstein tensor of the spacetime metric, and $K_{\Sigma}$ is the restriction of $K$ to $T_p\Sigma$ for $p\in \Sigma$ (likewise $W_{\Sigma}$, whereas $R_{\Sigma}$ is the scalar curvature of the induced metric $h$ on $\Sigma$). Expression \eqref{eq4.4} gives the linearized expansion $\theta$ for any spacelike codimension-$2$ surface $\Sigma$.

We now add $f$-terms to both sides of \eqref{eq4.4}. On the left-hand side we use
\begin{equation}
\label{eq4.5}
\begin{split}
D\theta_f=&\, D\theta-\varphi\nabla_{\nu}\nabla_{\ell} f = D\theta-\left [ \hess f (\nu,\ell) +\nabla_\nu\ell \cdot\nabla f\right ] \varphi.
\end{split}
\end{equation}
To simplify the right-hand side, we define
\begin{equation}
\label{eq4.6}
R^{f,m}:=R+2\Box f -\frac{(m+1)}{m}|\nabla f|^2
\end{equation}
and
\begin{equation}
\label{eq4.7}
G_{\mu\nu}^{f,m}:=R_{\mu\nu}^{f,m} -\frac12 R^{f,m} g_{\mu\nu}.
\end{equation}
If $G^{f,m}(u,v)\ge 0$ for all future-causal vectors $u$ and $v$ and some function $f$ and number $m$, we will say that the \emph{$m$-Bakry--\'Emery dominant energy condition holds}. Note that $R^{f,m}$ is not the trace of $\ric^{f,m}$ (except when $f=const$).

Next, we collect some of the expressions we will need. First, using $\ric=\ric^{f,m}-\hess f +\frac{1}{m}df\otimes df$ and choosing $w=\frac12 (\ell+k)$, we get that
\begin{equation}
\label{eq4.8}
\ric(w,\ell) = \ric^{f,m}(w,\ell) -\frac12 \hess f (k,\ell)-\frac12\hess (\ell,\ell) +\frac{1}{2m}\nabla_kf\nabla_{\ell} f +\frac{1}{2m}\left(\nabla_{\ell} f\right)^2.
\end{equation}
Next, we expand $\Box f$ using
\begin{equation}
\label{eq4.9}
\Box f =\Delta_{\Sigma}f-\left ( \hess f\right )(k,\ell ).
\end{equation}
This allows us to write that
\begin{equation}
\label{eq4.10}
\begin{split}
R=&\, R^{f,m}-2\Box f +\frac{(m+1)}{m}|\nabla f|^2\\
=&\, R^{f,m}-2\Delta_{\Sigma}f+2\hess f (k,\ell)-\frac{(m+1)}{m}\nabla_k f\nabla_{\ell}f+\frac{(m+1)}{m}\left \vert\nabla_{\Sigma}f\right\vert^2.
\end{split}
\end{equation}
Using \eqref{eq4.8} and \eqref{eq4.10}, then
\begin{equation}
\label{eq4.11}
\begin{split}
G(w,\ell) =&\,  \ric(w,\ell)-\frac12 R g(w,\ell) = \ric(w,\ell)+\frac12 R\\
=&\, G^{f,m}(w,\ell) +\frac12 \hess f (k-\ell,\ell) -\Delta_{\Sigma} f -\frac12\nabla_kf\nabla_{\ell} f +\frac{1}{2m}\left(\nabla_{\ell} f\right)^2\\
&\, +\frac{(m+1)}{2m}\left \vert\nabla_{\Sigma}f\right\vert^2.
\end{split}
\end{equation}
Then \eqref{eq4.4}, \eqref{eq4.5}, and \eqref{eq4.11} yield
\begin{equation}
\label{eq4.12}
\begin{split}
D\theta_f=&\, -\Delta_{\Sigma} \varphi +2W_\Sigma\cdot \nabla_\Sigma \varphi+\bigg[ \divergence_{\Sigma} W_{\Sigma} -\left \vert W_{\Sigma}\right \vert^2 +\frac12 R_{\Sigma} \\
&\,-G^{f,m}(w,\ell)-\frac12 \hess f (k-\ell,\ell) +\Delta_{\Sigma} f +\frac12\nabla_kf\nabla_{\ell} f -\frac{1}{2m}\left(\nabla_{\ell} f\right)^2 \\
&\, \left . -\frac{(m+1)}{2m}\left \vert\nabla_{\Sigma}f\right\vert^2 -\frac12 \left \vert K_{\Sigma} +A\right \vert^2 +\frac12 \theta\left ( \kappa +2K(\nu,\nu)\right )\right . \\
&\,  -\hess f (\nu,\ell) - \nabla_\nu\ell \cdot\nabla f\bigg ]\varphi,
\end{split}
\end{equation}
where once again each term in this expression is intended to be evaluated on $\Sigma$. Using
\begin{equation}
\label{eq4.13}
R^{f,m}_{\Sigma}:=R_{\Sigma}+2\Delta_{\Sigma}f-\frac{(m+1)}{m}\left \vert \nabla_{\Sigma} f\right \vert^2,
\end{equation}
and recognizing that the two Hessian terms cancel in \eqref{eq4.12}, we obtain
\begin{equation}
\label{eq4.14}
\begin{split}
D\theta_f
=&\, -\Delta_{\Sigma} \varphi +2W_\Sigma\cdot \nabla_\Sigma \varphi +\bigg [ \divergence_{\Sigma} W_{\Sigma} -\left \vert W_{\Sigma}\right \vert^2 +\frac12 R_{\Sigma}^{f,m} \\
&\,-G^{f,m}(w,\ell)+\frac12\nabla_kf\nabla_{\ell} f -\frac{1}{2m}\left(\nabla_{\ell} f\right)^2 -\frac12 \left \vert K_{\Sigma} +A\right \vert^2\\
&\, +\frac12 \theta\left ( \kappa +2K(\nu,\nu)\right ) - \nabla_\nu\ell \cdot\nabla f\bigg]\varphi .
\end{split}
\end{equation}

Following \cite{Galloway4}, we write the mean curvature of the hypersurface $V$ as
\begin{equation}
\label{eq4.15}
\tau:=\tr_V K= \frac12\left ( \theta+\kappa\right ) +K(\nu,\nu),
\end{equation}
which allows us to rewrite $ \kappa +2K(\nu,\nu)$ as
\begin{equation}
\label{eq4.16}
\begin{split}
\kappa +2K(\nu,\nu)=&\, 2\tau -\theta\\
=&\, 2\tau -\theta_f-\nabla_{\ell}f.
\end{split}
\end{equation}
We have
\begin{equation}
\label{eq4.17}
\begin{split}
D\theta_f=&\, -\Delta_{\Sigma} \varphi +2W_\Sigma\cdot \nabla_\Sigma \varphi+\bigg[ \divergence_{\Sigma} W_{\Sigma} -\left \vert W_{\Sigma}\right \vert^2 +\frac12 R_{\Sigma}^{f,m}  \\
&\,-G^{f,m}(w,\ell) -\frac12 \left \vert K_{\Sigma} +A\right \vert^2 + \theta_f \left ( \tau -\frac12\theta_f -\nabla_{\ell}f\right ) \\
&\,   +\left ( \tau +\frac12\nabla_kf -\frac{(m+1)}{2m}\nabla_{\ell}f\right ) \nabla_{\ell} f -\nabla_\nu\ell \cdot\nabla f \bigg]\varphi .
\end{split}
\end{equation}

Now we assume that $\Sigma$ is $f$-minimal; i.e., $\theta_f=0$, and that $\theta_f$ is stable (at linear order) against perturbations of $\Sigma$ (since its outermost by assumption). That is, we require that for any smooth function $\varphi:\Sigma\to {\mathbb R}$, we have
\begin{equation}
\label{eq4.18}
\begin{split}
0\le &\, L^f\varphi\\
:=&\, -\Delta_{\Sigma} \varphi +2W_\Sigma\cdot \nabla_\Sigma \varphi +\bigg [ \divergence_{\Sigma} W_{\Sigma} -\left \vert W_{\Sigma}\right \vert^2 +\frac12 R_{\Sigma}^{f,m}  \\
&\, -G^{f,m}(w,\ell) -\frac12 \left \vert K_{\Sigma} +A\right \vert^2
+\left ( \tau +\frac12\nabla_kf -\frac{(m+1)}{2m}\nabla_{\ell}f\right )\nabla_\ell f\\
&\,  -\nabla_\nu\ell \cdot\nabla f \bigg ]\varphi .
\end{split}
\end{equation}

\subsection{A special case}\label{section4.2}
Now we further assume that
\begin{equation}
\label{eq4.19}
\nabla_{\ell}f\vert_{\Sigma}=0.
\end{equation}
See Section \ref{section5} and Lemma \ref{lemma5.5} for a discussion of this assumption. Then \eqref{eq4.17} reduces to
\begin{equation}
\label{eq4.20}
\begin{split}
0\le &\, -\Delta_{\Sigma} \varphi +Df\cdot D\varphi - \varphi K(Df,\nu) +2 W_{\Sigma}\cdot D\varphi +\bigg [ \divergence_{\Sigma} W_{\Sigma} -\left \vert W_{\Sigma}\right \vert^2   \\
&\, +\frac12 R_{\Sigma}^{f,m} -G^{f,m}(w,\ell) -\frac12 \left \vert K_{\Sigma} +A\right \vert^2 \bigg ]\varphi ,
\end{split}
\end{equation}
where we used Lemma \ref{lemmaC.1} in Appendix \ref{appendixC} to replace the $-\nabla_\nu\ell \cdot\nabla f$ term . Also, we let $D := \nabla_\Sigma$ to denote the gradient with respect to the metric on $\Sigma$ (since $D$ is no longer used to denote the variation of $\theta_f$).

From here on, the calculation is almost standard.  We will eliminate the first-derivative term $2\left \langle W_{\Sigma},D\varphi\right \rangle_{\Sigma}$ in favour of a term with definite sign. To do this, write $\varphi=e^u$. Then
\begin{equation}
\label{eq4.21}
\begin{split}
&\, -\Delta_{\Sigma}\varphi+Df\cdot D\varphi +2W_{\Sigma}\cdot D\varphi\\
=&\, e^u\left [ -\Delta_{\Sigma} u +Df\cdot Du -\left \vert W_{\Sigma}-Du\right \vert^2 +\left \vert W_{\Sigma}\right \vert^2 \right ],
\end{split}
\end{equation}
so \eqref{eq4.20} becomes
\begin{equation}
\label{eq4.22}
\begin{split}
0\le &\, -\Delta_{\Sigma}u +Df\cdot Du -K(Df, \nu)+\divergence_{\Sigma} W_{\Sigma} -\left \vert W_{\Sigma}-Du\right \vert^2 +Q\\
=&\, \divergence_{\Sigma}\left ( W_{\Sigma} - Du \right ) +Du\cdot Df -K(Df,\nu)-\left \vert W_{\Sigma}-Du\right \vert^2 +Q,
\end{split}
\end{equation}
where
\begin{equation}
\label{eq4.23}
Q:=\frac12R_{\Sigma}^{f,m} -G_{\mu\nu}^{f,m}\ell^{\mu}w^{\nu}-\frac12 \left \vert K_{\Sigma} +A\right \vert^2 .
\end{equation}

Multiplying \eqref{eq4.22} by $\psi^2$ for some $\psi\in C^{\infty}(\Sigma)$, we obtain
\begin{equation}
\label{eq4.24}
\begin{split}
0\le &\, \divergence_{\Sigma}\left ( \psi^2\left ( W_{\Sigma}-Du\right ) \right ) -2\psi\left ( W_{\Sigma}-Du\right )\cdot D\psi +\psi^2 Du\cdot Df\\
&\, -\psi^2 K(Df,\nu) -\left \vert W_{\Sigma}-Du\right \vert^2_{\Sigma}\psi^2+Q\psi^2\\
\le &\,  \divergence_{\Sigma}\left ( \psi^2\left ( W_{\Sigma}-Du\right ) \right ) +2|\psi| \left \vert W_{\Sigma}-Du\right \vert  |D\psi|+\psi^2 Du\cdot Df\\
&\, -\psi^2 K(Df,\nu) -\left \vert W_{\Sigma}-Du\right \vert^2_{\Sigma}\psi^2+Q\psi^2\\
= &\, \divergence_{\Sigma}\left ( \psi^2\left ( W_{\Sigma}-Du\right ) \right ) -\left ( |D\psi | -\left \vert W_{\Sigma}-Du\right \vert |\psi |\right )^2\\
&\, +\psi^2 Du\cdot Df - \psi^2 K(Df,\nu)+|D\psi |^2 +Q\psi^2 \\
\le &\, \divergence_{\Sigma}\left ( \psi^2\left ( W_{\Sigma}-Du\right ) \right )+\psi^2 Du\cdot Df - \psi^2 K(Df,\nu)+|D\psi |^2 +Q\psi^2
\\
= &\, \divergence_{\Sigma,f}\left (\psi^2 \left ( W_{\Sigma}  - Du\right )\right ) + |D\psi|^2 + Q\psi^2,
\end{split}
\end{equation}
where $\divergence_{\Sigma,f}(X) := \divergence_\Sigma(X) - Df\cdot X$ for vector fields $X$ on $\Sigma$. In the last equality, we used the fact that $K(Df,\nu) = Df \cdot W_\Sigma$ (which follows immediately from the definition of $W_{\Sigma}:=K(\cdot,\nu)$).

Multiplying \eqref{eq4.24} by $e^{-f}$ and integrating over $\Sigma$, we obtain
\begin{equation}
\label{eq4.25}
\begin{split}
\int_{\Sigma} ( |D\psi |^2 +Q\psi^2 )e^{-f}dV_{\Sigma} &\ge 0.
\end{split}
\end{equation}

Recall that $\Delta_{\Sigma,f} \psi = \divergence_{\Sigma,f}(D \psi) = \divergence_\Sigma(D\psi) - Df\cdot D\psi$. Consider the weighted eigenvalue problem:
\begin{equation}
\label{eq4.26}
\begin{split}
\mu \psi = &\, -\Delta_{\Sigma,f}\psi +Q\psi \\
1=&\, \int_{\Sigma} \psi^2 e^{-f}dV_{\Sigma},
\end{split}
\end{equation}
on a compact manifold $\Sigma$. The Rayleigh quotient for \eqref{eq4.26} is given by the left-hand-side of \eqref{eq4.25}. Hence we have the following result.

\begin{lemma}\label{lemma4.1}
The principal eigenvalue $\mu_1$ for the problem \eqref{eq4.26} is nonnegative.
\end{lemma}

We can now prove our main results. The proofs follow closely the proof in \cite{GS}, but for the reader's convenience we supply details.

\begin{proof}[Proof of Theorem \ref{theorem1.2}]
Now set $n := \dim \Sigma\ge 2$. Let $\psi_1>0$ be an eigenfunction for the principal eigenvalue $\mu_1$ of \eqref{eq4.26}. By standard theory (see e.g. \cite[Thm. A.10]{Lee}), we can choose $\psi_1 > 0$. The conformal transformation
\begin{equation}
\label{eq4.28}
{\tilde h} = \psi_1^{2p}h,
\end{equation}
applied to the metric $h$ on $\Sigma$, produces a metric ${\tilde h}$ whose scalar curvature is
\begin{equation}
\label{eq4.29}
R_{\Sigma}(\tilde h) = \psi_1^{-2p}\left \{ R_{\Sigma}(h)-2(n-1)p\frac{\Delta_h\psi_1}{\psi_1}  + (n-1)p\left [ 2-(n-2)p\right ] \frac{\left \vert D \psi_1\right \vert^2_h}{\psi_1^2}\right \}.
\end{equation}
Here we have replaced the notation $\Delta_{\Sigma}$ by $\Delta_h$ to distinguish the Laplacian on $\Sigma$ defined by $h$ from that defined by ${\tilde h}$. Using its conformal transformation property, we can write the scalar Laplacian $\Delta_{\tilde h}$ associated to ${\tilde h}$ in terms of $\Delta_h$, yielding
\begin{equation}
\label{eq4.30}
\Delta_{\tilde h} f = \psi_1^{-2p}\left [ \Delta_h f +(n-2)p\frac{h(D\psi_1,Df)}{\psi_1}\right ].
\end{equation}
Finally, we replace $f$ by
\begin{equation}
\label{eq4.31}
{\tilde f}:=f +k\log \psi_1,
\end{equation}
so that
\begin{equation}
\label{eq4.32}
\begin{split}
\Delta_{\tilde h} {\tilde f} =&\, \Delta_{\tilde h} f +k\left \{ \frac{\Delta_{\tilde h}\psi_1}{\psi_1} -\frac{\left \vert \tilde{D}\psi_1\right \vert_{\tilde h}^2}{\psi_1^2}\right \} \\
=&\, \psi_1^{-2p}\left \{ \Delta_h f +(n-2)p\frac{h(D\psi_1,Df)}{\psi_1}+k \frac{\Delta_h\psi_1}{\psi_1}\right .\\
&\, \left . +k\left [ (n-2)p-1\right ] \frac{\left \vert D\psi_1\right \vert_h^2}{\psi_1^2}\right \}
\end{split}
\end{equation}

Putting things together, we obtain
\begin{equation}
\label{eq4.33}
\begin{split}
R_{\Sigma}^{{\tilde f},m}({\tilde h}) := &\, R_{\Sigma}({\tilde h}) +2\Delta_{\tilde h}{\tilde f} -\frac{(m+1)}{m}\left \vert \tilde{D}{\tilde f}\right \vert_{\tilde h}^2\\
= &\, \psi_1^{-2p}\left \{ R_{\Sigma}(h)-2(n-1)p\frac{\Delta_h\psi_1}{\psi_1} + (n-1)p\left [ 2-(n-2)p\right ] \frac{\left \vert D \psi_1\right \vert_h^2}{\psi_1^2} \right . \\
&\, \left . + 2\Delta_h f +2(n-2)p\frac{h(D\psi_1,Df)}{\psi_1}+2k \frac{\Delta_h\psi_1}{\psi_1}\right . \\
&\, \left . +2k\left [ (n-2)p-1\right ] \frac{\left \vert D\psi_1\right \vert_h^2}{\psi_1^2} \right .\\
&\, \left . -\frac{(m+1)}{m} \left [ |Df|_h^2 +2k\frac{h(Df,D\psi_1)}{\psi_1} +k^2 \frac{\left \vert D\psi_1 \right\vert_h^2}{\psi_1^2}\right ] \right \} .
\end{split}
\end{equation}

Now simultaneously solve for $p$ and $k$
\begin{equation}
\label{eq4.34}
\begin{split}
1=&\, (n-1)p-k,\\
1=&\, (n-2)p-\frac{(m+1)}{m}k
\end{split}
\end{equation}
to obtain a coefficient of $-2$ for the $\frac{\Delta_h\psi_1}{\psi_1}$ term and $+2$ for $\frac{h(Df,D\psi_1)}{\psi_1}$ term. We obtain
\begin{equation}
\label{eq4.35}
\begin{split}
p=&\, \frac{1}{n+m-1},\\
k=&\,-\frac{m}{n+m-1}.
\end{split}
\end{equation}
Then those two terms combine to give $-\frac{2}{\psi_1}\Delta_{h,f}\psi_1$.
Therefore
\begin{equation}
\label{eq4.36}
\begin{split}
R_{\Sigma}^{{\tilde f},m}({\tilde h}) = &\, \psi_1^{-2p}\left \{ R_{\Sigma}(h) + 2\Delta_h f -\frac{(m+1)}{m} |Df|_h^2-2\frac{\Delta_{h,f}\psi_1}{\psi_1}\right . \\
&\, \left . + \left ( \frac{n+m}{n+m-1}\right ) \frac{\left \vert D \psi_1\right \vert_h^2}{\psi_1^2} \right \}\\
=&\, \psi_1^{-2p}\left \{ R_{\Sigma}^{f,m}(h)-2\frac{\Delta_{h,f}\psi_1}{\psi_1}+\left ( \frac{n+m}{n+m-1}\right ) \frac{\left \vert D \psi_1\right \vert_h^2}{\psi_1^2} \right \}\\
=&\, \psi_1^{-2p}\left \{ R_{\Sigma}^{f,m}(h)+2\mu_1-2Q+\left ( \frac{n+m}{n+m-1}\right ) \frac{\left \vert D \psi_1\right \vert_h^2}{\psi_1^2} \right \}\\
\ge &\, \psi_1^{-2p}\left \{ 2\mu_1+\left ( \frac{n+m}{n+m-1}\right ) \frac{\left \vert D \psi_1\right \vert_h^2}{\psi_1^2} \right \}
\\
\geq & 0.
\end{split}
\end{equation}
We used \eqref{eq4.26} to get the third equality. The first inequality follows from \eqref{eq4.23}. The second inequality follows since $\mu_1 \geq 0$ by Lemma \ref{lemma4.1}.
\end{proof}

\section{Warped products}\label{section5}
\setcounter{equation}{0}

\subsection{The proof of Theorem \ref{theorem1.9}}
Consider the warped product metric
\begin{equation}
\label{eq5.1}
\gamma= g\oplus e^{-2f/m}h,
\end{equation}
on the manifold $N=M\times S$, where $g$ is a Lorentzian metric on an $n$-dimensional spacetime $M$, $h$ is a Riemannian metric on the compact $m$-manifold $S$, and $f:M\to{\mathbb R}$. More precisely, let $\iota:M\to N$ and $\zeta:S\to N$ be inclusions such that $p=\iota(q)=\zeta(r)\in N$. Then the decomposition $T_pN=\iota_*T_q M\oplus \zeta_* T_rS$ is orthogonal, with $\iota_*T_q M$ in the kernel of $h$ and $\zeta_*T_rS$ in the kernel of $g$. We use $\oplus$ to emphasize the direct sum (sometimes called the direct product) where it seems helpful to do so, but we will usually suppress inclusions (and their pullbacks/pushforwards) to lessen the burden of notation. Note that $(M,g)$ is totally geodesic in $(N,\gamma)$, while $(S,h)$ is umbilic.

Warped products occur in the physics of Kaluza-Klein type compactifications, where the weighting $f$ represents a dilaton (perhaps after field redefinition). In the warped product formulation, the electromagnetic field would appear in the stress-energy tensor and be subject to any assumed energy conditions. (In a full implementation of the Kaluza-Klein idea, electromagnetism would be incorporated by twisting $\gamma$, which would yield equations for the electromagnetic field, making energy conditions redundant.)

A standard calculation yields
\begin{equation}
\label{eq5.2}
\begin{split}
\ric(N,\gamma) =&\, \ric^{f,m}(M,g) \oplus \left \{ \ric(S,h)+\frac{e^{-2f/m}}{m}\left [ \square_g f -\vert \nabla f \vert_g^2\right ] h \right \},\\
=&\, \ric^{f,m}(M,g) \oplus \left \{ \ric(S,h)+\frac{e^{-2f/m}}{m}\left ( \square_{g,f} f \right ) h \right \},\\
R(N,\gamma) = &\, R^{f,m}(M,g)+e^{2f/m}R(S,h),
\end{split}
\end{equation}
where $\ric(S,h)$ and $R(S,h)$ are taken to be zero if $\dim S =1$, and $\square_{g,f} f$ denotes the \emph{drift Laplacian} defined by
\begin{equation}
\label{eq5.3}
\square_{g,f} u=\square_g u -g(\nabla f,\nabla u).
\end{equation}

Recalling Theorem \ref{theorem1.2}, we make the following observation.
\begin{remark}\label{remark5.1}
If the $m$-Bakry--\'Emery scalar curvature $R^{f,m}(M,g)$ of the base $(M,g)$ and the scalar curvature $R(S,h)$ of the fibre are nonnegative, then the scalar curvature $R(N,\gamma)$ of the total space is nonnegative. If one of the former is positive and the other is nonnegative, then the scalar curvature $R(N,\gamma)$ of the total space is positive.
\end{remark}

\begin{proof}[Proof of Theorem \ref{theorem1.9}]
We prove the theorem for $\dim\Sigma = 3$. The two-dimensional case is similar.
The reverse implication is trivial by simply taking $f = const$. Conversely, suppose $\Sigma$ does not admit a metric of positive scalar curvature. By the classification of oriented prime three-manifolds, $\Sigma$ is diffeomorphic to either a quotient of $\mathbb{S}^3$, $\mathbb{S}^2\times \mathbb{S}^1$, or an aspherical space, and the first two are ruled out. Therefore $\Sigma$ is aspherical. Hence $\Sigma \times {\mathbb S}^1$ is also aspherical, and thus it also does not carry a metric of positive scalar curvature \cite[Theorem 2]{CL} (see also \cite{Gromov}). Therefore, by the second line of equation \eqref{eq5.2}, this proves the theorem for $m = 1$. For $0 < m < 1$, we have $0 < R^{f,m}(g)=R^{f,1}(g)+\left ( 1-\frac{1}{m}\right ) |df|^2 \le R^{f,1}(g)$, which proves the theorem for these values of $m$. Similar reasoning applies for $1 < m \leq 2$.
\end{proof}

\begin{proof}[Proof of Remark \ref{remark1.10}]\:
\begin{itemize}
\item [1.] See Claim \ref{claim5.2} below.
\item [2.] If $R^{f,m}>0$ on $\Sigma$ for some $m>0$ (possibly not an integer), then $R^{f,k}>0$ for any integer $k\ge m$, and then by the second line of \eqref{eq5.2} we have a metric of positive scalar curvature on $\Sigma\times {\mathbb T}^k$. But if $\Sigma$ is aspherical then so is $\Sigma\times {\mathbb T}^k$, so we would have a positive scalar curvature metric on an aspherical manifold, falsifying the conjecture.
\item[3.] In this nonorientable case, one uses the Projective Plane Theorem instead of the Sphere Theorem to deduce that $\Sigma$ is aspherical but one needs the additional hypothesis that $\Sigma$ does not contain any two-sided $\mathbb{RP}^2$ to obtain the conclusion in the remark. Although not needed for the proof, we note that
there are closed prime manifolds with a two-sided $\mathbb{RP}^2$ \cite{Negami} other than $\mathbb{S}^1 \times \mathbb{RP}^2$ (which admits a metric with positive scalar curvature). Moreover, it follows from
Proposition 2.2 in \cite{Swarup} and Theorem 5.1 in \cite{Li-Zhang} that there are many prime 3-manifolds with a two-sided $\mathbb{RP}^2$ which do not admit metrics of positive scalar curvature and yet are not aspherical.
\end{itemize}
\end{proof}

An interesting question is whether there are manifolds that admit positive $m$-Bakry--\'Emery scalar curvature for some $m>0$ but do not admit positive scalar curvature. With that in mind, we make the following observation.

\begin{claim}\label{claim5.2} \:
\begin{itemize}
\item [(i)] For $\Sigma$ a closed oriented manifold, if there is a function $f$ and metric $g$ such that $R^{f,m}(g)>0$ for some $m\in (0,1]$, then $\Sigma$ admits a metric of positive scalar curvature.
\item [(ii)] If $\dim \Sigma=3$ or $\dim \Sigma=5$, $\Sigma$ will admit a metric of positive scalar curvature if there is a function $f$ and metric $g$ such that $R^{f,m}(g)>0$ for some $m\in (0,2]$.
\end{itemize}
\end{claim}

As above, since $R^{f,m}>0$ implies that $R^{f,k}>0$ for any $k\ge m$, without loss of generality we round $m$ up to the nearest integer, allowing us to consider the warped product of $\Sigma$ with an $m$-torus. Then to justify part (i) of the claim, we need only consider the case of $m=1$. Recently Rosenberg and Xu \cite{RX} have announced a proof of the ${\mathbb S}^1$-stability conjecture under certain assumptions. Their work applies whenever the tangent vectors to the ${\mathbb S}^1$ factor in the topological product $\Sigma\times {\mathbb S}^1$ lie within the cones (defined by a positive scalar curvature metric) whose axis at each point is the normal vector (with orientation chosen to match that defined by the ${\mathbb S}^1$ factor) to the $\Sigma$ factor. When the positive scalar curvature metric on $\Sigma\times {\mathbb S}^1$ is a warped product, this condition is trivially met since then these two vector fields are proportional and then $\Sigma$ would admit a positive scalar curvature metric.

To justify part (ii) of the claim, when $m\in (1,2]$ we can round up to $m=2$. If $\dim \Sigma=3$, we obtain from equations \eqref{eq5.1} and \eqref{eq5.2} a positive scalar curvature warped product metric on the $5$-manifold $\Sigma\times {\mathbb T}^2$. Applying the result of Rosenberg and Xu \cite{RX} to this warped product, we obtain a positive scalar curvature metric (which might not be a warped product) on the $4$-manifold $\Sigma\times {\mathbb S}^1$. But the ${\mathbb S}^1$ conjecture is true for $\Sigma$ a $3$-manifold \cite[Corollary 7.34]{Chodosh}, so there must be a positive scalar curvature metric on $\Sigma$. If $\dim \Sigma=5$, the proof is even simpler since it is known that if a closed $7$-manifold of the form $\Sigma^6\times {\mathbb S}^1$ admits a positive scalar curvature metric (whether in the form of a warped product or not) then so does $\Sigma^6$, and if $\Sigma^6\simeq \Sigma^5\times{\mathbb S}^1$ admits a positive scalar curvature metric then so does $\Sigma^5$ \cite[Remark 2.26]{Rade}.

\begin{remark}\label{remark5.3}
We see from this that the result of \cite{RX} would imply another proof of part of Theorem \ref{theorem1.9}.
\end{remark}

\subsection{Interpretation of the dominant energy condition and trapping}
Continuing from \eqref{eq5.2}, we now compute that the Einstein tensor $G(N,\gamma)$ of the total space becomes
\begin{equation}
\label{eq5.4}
\begin{split}
G(N,\gamma):=&\, \ric(N,\gamma) -\frac12 R(N,\gamma)\gamma\\
= &\, \left \{ G^{f,m}(M,g) -\frac12 e^{2f/m}R(S,h)g\right \}\\
&\, \oplus \left \{ G(S,h)+e^{-2f/m}\left [ \frac{1}{m}\left ( \square_gf-|\nabla f|_g^2\right )-\frac12 R^{f,m}(M,g)\right ] h \right \}.
\end{split}
\end{equation}

\begin{lemma}[Warped product interpretation of DEC]\label{lemma5.4}
Let $U,V\in \iota_*T_q M$ be future-causal and $m$ is a positive integer. Let $(N,\gamma)$ be a warped product spacetime with $\gamma$ defined in \eqref{eq5.1}.
\begin{itemize}
\item [(i)] If $ G^{f,m}(M,g)(U,V)\ge 0$ and $R(S,h)\ge 0$, then $G(N,\gamma)(U,V)\ge 0$.
\item [(ii)] If $ G(N,\gamma)(U,V)\ge 0$ and $R(S,h)\le 0$, then $G^{f,m}(M,g)(U,V)\ge 0$.
\end{itemize}
\end{lemma}

\begin{proof}
Immediate from \eqref{eq5.4}.
\end{proof}
Then for $m$ a positive integer, we can view the $m$-Bakry-\'Emery dominant energy condition (DEC) as a weak form of the dominant energy condition in a warped product manifold which is only applied to null vectors orthogonal to the fibres, provided the scalar curvature of the fibres has a sign.

To complete the picture, we turn our attention to null mean curvatures. It's convenient to use (abstract) index notation, with greek letters (e.g., $\mu$) denoting tensors on $N$, middle latin letters (e.g., $i$) denoting tensors on $M$ (pushed forward to $N$), and early latin letters (e.g., $a$) denoting tensors on $S$. Let $\ell$ be a null vector in $T_pN$. Then its null mean curvature is easily computed from $\gamma$. In this notation, $\nabla_{\mu}$ denotes the Levi-Civita connection of $(N,\gamma)$, $\nabla_i$ denotes the Levi-Civita connection of $(M,g)$, and $\nabla_a$ denotes the Levi-Civita connection of $(S,h)$. Then for $\ell$ an element of $T_p^*N$, we have
\begin{equation}
\label{eq5.5}
\left [ \nabla_{\mu}\ell_{\nu}\right ]=\left [ \begin{array}{cc} \nabla_i\ell_j & \nabla_i \ell_a +\frac{1}{m} \ell_a \nabla_i f\\ & \\
\nabla_a \ell_i+\frac{1}{m} \ell_a \nabla_i f & \nabla_a\ell_b -\frac{1}{m}e^{-2f/m}\left ( \nabla_{\ell} f\right ) h_{ab}\end{array} \right ].
\end{equation}
If $\ell$ is (metric-dual to) a null vector that lies in $i_* T_qN$, then $\ell_a=0$ and the null mean curvature $\theta^{(N,\gamma)}$ of $\ell$ in $N$ is
\begin{equation}
\label{eq5.6}
\theta^{(N,\gamma)}=\theta^{(M,g)}-g(\ell,\nabla f) = \theta^{(M,g)}-\nabla_{\ell}f = \theta_f,
\end{equation}
where $\theta^{(M,g)}$ is the null mean curvature of $\ell$ in $M$ and $\nabla$ can be taken to be the Levi-Civita connection on $(M,g)$. We have the following interpretation.

\begin{lemma}\label{lemma5.5}
Let $\ell$ be future null and normal to a codimension-$2$ spacelike hypersurface in the warped product $(N,\gamma)$, where $\gamma$ is defined in \eqref{eq5.1}. Then $\Sigma$ is future trapped (marginally trapped) if and only if $\theta_f<0$ (respectively, $\theta_g=0$). As well, we have the following.
\begin{itemize}
\item [(i)] If $\Sigma$ is trapped in $(M,g)$ and $\nabla_{\ell}f\ge 0$ then $\Sigma$ is future trapped in $(N,\gamma)$.
\item [(ii)] If $\Sigma$ is trapped in $(N,\gamma)$ and $\nabla_{\ell}f\le 0$ then $\Sigma$ is future trapped in $(M,g)$.
\item [(iii)] ``Trapped'' can be replaced with ``marginally trapped'' in these statements if also the inequality on $\nabla_{\ell} f$ is replaced by $\nabla_{\ell} f=0$.
\end{itemize}
\end{lemma}
\begin{proof}
Immediate from \eqref{eq5.6}.
\end{proof}

\appendix
\section{Flowing the Bakry--\'Emery scalar curvature}\label{appendixA}
\setcounter{equation}{0}

\noindent Consider a smooth one-parameter family of metrics $g_t$ and functions $f_t$, $t\ge 0$, such that $g=g_0$ and $f=f_0$. The question then is to compute the variation with $t$ of the $m$-Bakry--\'Emery scalar curvature.  Straightforward calculations show that
\begin{equation}
\label{eqA.1}
\frac{\partial R}{\partial t}= -R^{ij}\frac{\partial g_{ij}}{\partial t} +\nabla^i\nabla^j \frac{\partial g_{ij}}{\partial t}-\Delta\left ( g^{ij}\frac{\partial g_{ij}}{\partial t}\right ) ,
\end{equation}
\begin{equation}
\label{eqA.2}
\begin{split}
\frac{\partial}{\partial t} (2\Delta f) =&\, -2\frac{\partial g_{ij}}{\partial t}\nabla^i\nabla^j f- \left [ \nabla^i f \cdot \nabla^j \frac{\partial g_{ij}}{\partial t} + \nabla^j f \cdot \nabla^i \frac{\partial g_{ij}}{\partial t} \right . \\
&\, \left . -g^{ij}\nabla_{\nabla f} \frac{\partial g_{ij}}{\partial t}\right ] +2\Delta\frac{\partial f}{\partial t},
\end{split}
\end{equation}
and
\begin{equation}
\label{eqA.3}
\frac{\partial}{\partial t} \left [ -\frac{(m+1)}{m}|df|^2\right ] = \frac{(m+1)}{m}\frac{\partial g_{ij}}{\partial t} \nabla^i f \cdot \nabla^j f-\frac{2(m+1)}{m}\nabla^i f \cdot \nabla_i \frac{\partial f}{\partial t}.
\end{equation}
Adding these three expressions, we get that
\begin{equation}
\label{eqA.4}
\begin{split}
\frac{\partial R^{f,m}}{\partial t} = &\, \left [ \nabla^i\nabla^j -g^{ij}\Delta -\nabla^i f \cdot \nabla^j -\nabla^j f \cdot \nabla^i +g^{ij}\nabla_{\nabla f} -\left ( R^{f,m}\right )^{ij}\right ] \frac{\partial g_{ij}}{\partial t}\\
&\, +2\left [ \Delta -\frac{(m+1)}{m}\nabla^i f \cdot \nabla_i\right ] \frac{\partial f}{\partial t} .
\end{split}
\end{equation}

\noindent \begin{lemma}\label{lemmaA.1}
If the metric $g$, function $f$, and $m\in {\mathbb R}\cup\{ \infty\}$ are such that $R^{f,m}(g)\ge 0$ on a closed manifold $\Sigma$ of dimension $\geq 2$, and if $R^{f,m}(g)$ is not identically zero, then there is a metric ${\hat g}$ such that $R^{f,m}({\hat g})> 0$ on $\Sigma$.
\end{lemma}

\begin{proof}
Consider \eqref{eqA.4} evaluated along the curve of metrics and functions defined by integrating
\begin{equation}
\label{eqA.5}
\frac{\partial g_t}{\partial t} = -R^{f,m}(g_t) g_t,\quad g_0=g,\quad f_t=f_0=f.
\end{equation}
Then $\frac{\partial f}{\partial t}=0$ and then \eqref{eqA.4} yields
\begin{equation}
\label{eqA.6}
\frac{\partial R^{f,m}}{\partial t} =(n-1)\Delta R^{f,m} -(n-2)\nabla_{\nabla f} R^{f,m} +\left ( R^{f,m}\right )^2.
\end{equation}
This equation is uniformly parabolic, hence the flow exists for $t\in [0,T)$ with some $T\in (0,\infty]$. Given $R^{f,m}(p,0)\ge 0$ for all $p\in \Sigma$, the weak maximum principle yields that $R^{f,m}(p,t)\ge 0$ for all $p\in \Sigma$ and all $t\ge 0$ in the domain of the flow defined by \eqref{eqA.5}. But then the strong maximum principle \cite[Theorem 6.54, pp 245--246]{CLN} (see \cite[Proposition 1.1, p 416]{Taylor} or \cite[Theorem 4, p 172]{PW} for proofs) asserts that either $R^{f,m}(p,t)> 0$ for all $p\in \Sigma$ and all $t\in [0,T)$, and then the desired metric ${\hat g}$ is any of the metric $g_t$ with $t>0$, or $R^{f,m}(\cdot,\cdot)$ is identically zero. The latter is only possible if $R^{f,m}(\cdot,0)$ is identically zero.
\end{proof}

\section{Conformal transformations}\label{appendixB}
\setcounter{equation}{0}

\noindent If
\begin{equation}
\label{B.1}
{\tilde g}=e^{-2af/(n-2)}g,\quad f\in C^{\infty}(M),\quad a\in {\mathbb R},
\end{equation}
then
\begin{equation}
\label{eqB.2}
\ric_{\tilde g}=\ric_g+a\hess_g f +\frac{a^2}{n-2} df\otimes df +\frac{a}{n-2}\left [ \Delta_g f -a|df|_g^2\right ] g.
\end{equation}
For any $u\in C^{\infty}(M)$, we have
\begin{equation}
\label{eqB.3}
\hess_{\tilde g} u = \hess_g u +\frac{a}{n-2}\left [ df\otimes du + du\otimes df -g(\nabla f,\nabla u)g\right ].
\end{equation}

Let $u=(1-a)f$, $a\neq 1$. Then we obtain from the above formulas that
\begin{equation}
\label{eqB.4}
\begin{split}
\ric_{\tilde g} +\hess_{\tilde g} u =&\, \ric_g+\hess_g f - \frac{a(a-2)}{(n-2)}df\otimes df + \frac{a}{n-2}\left [ \Delta_g f -|df|_g^2\right ]g\\
=&\, \ric_g+\hess_g f - \frac{a(a-2)}{(n-2)}df\otimes df + \frac{a}{n-2}\left ( \Delta_{g,f}f\right ) g,
\end{split}
\end{equation}
where $\Delta_{g,f}$ denotes the drift Laplacian defined in \eqref{eq5.3}. Then
\begin{equation}
\label{eqB.5}
\begin{split}
\ric_{\tilde g}^{u,{\tilde m}}=&\, \ric_g+\hess_g f - \left [ \frac{(1-a)^2}{\tilde m} +\frac{a(a-2)}{n-2}\right ] df\otimes df
+ \frac{a}{n-2}\left ( \Delta_{g,f}f\right ) g \\
=&\, \ric_g^{f,m} + \frac{a}{n-2}\left ( \Delta_{g,f}f\right ) g,
\end{split}
\end{equation}
where
\begin{equation}
\label{eqB.6}
\frac{1}{m}:= \frac{(1-a)^2}{\tilde m} +\frac{a(a-2)}{n-2} .
\end{equation}
It is not hard to see that for any pair $m,{\tilde m}\in (-\infty,2-n)\cup (0,\infty)$, there is a choice of $a$ that will map ${\tilde m}$ to $m$ via the formula \ref{eqB.6}. The relation can be inverted to map any $m$ in this domain to ${\tilde m}$ in the same domain. Note that of course such a choice of $a$ also fixes the conformal transformation relating $g$ and ${\tilde g}$ and the scale transformation relating $f$ and $u$.

We now consider the singular cases. The first such case arises for $a=2$, and then $m={\tilde m}$. The next singular case is when $m=0$, which means that $f=const$, so $u=(1-a)f=const$ unless $a=1$, and so if $a\neq 0$ then we may take ${\tilde m}=0$. Likewise if we begin from ${\tilde m}=0$ and choose $a\neq 1$, we conclude that we can set $m=0$.

Finally, we consider $a=1$. Then $u=0$, so we can return to equation \eqref{eqB.2}. It yields
\begin{equation}
\label{eqB.7}
\begin{split}
\ric_{\tilde g}=&\, \ric_g+\hess_g f +\frac{df\otimes df}{n-2} +\frac{\Delta_{g,f} f}{n-2} g\\
=&\, \ric_g^{f,2-n} +\frac{\Delta_{g,f} f}{n-2} g .
\end{split}
\end{equation}

\section{Evaluation of $\nabla_{\nu}\ell$}\label{appendixC}
\setcounter{equation}{0}
In this section, we evaluate the term $\nabla_\nu \ell \cdot \nabla f$ which first appears in equation \eqref{eq4.5}. We will substitute this evaluation into equation\eqref{eq4.20}.

Recall the set up: $V$ is a spacelike hypersurface in a weighted spacetime $(M,g,f)$ with future-directed unit normal $w$. $\Sigma$ is a closed two-sided hypersurface in $V$ with normal $\nu$ on $\Sigma$ denoting the outward direction. For $\varphi \in C^\infty(\Sigma)$, we consider variations $\Sigma_t$ of $\Sigma$ within $V$ given by
\begin{equation}
\label{eqC.1}
\Sigma\ni p \mapsto \exp_p\left ( t\varphi(p)\nu(p)\right ) \in V,
\end{equation}
where $\exp_p$ is the exponential map derived by the Riemannian metric on $V$.

Let $X$ be the vector field on $V$ whose integral curves are precisely given by \eqref{eqC.1}. Note that
\begin{equation}
\label{eqC.2}
X|_\Sigma = \varphi \nu.
\end{equation}
Set $\ell = w + \nu$. Then $\left \{ \ell, \nu, \partial_3, \dotsc, \partial_n\right \}$ is a basis for $TM$ at $\Sigma$ where $\left \{ \partial_3, \dotsc, \partial_n\right \}$ is a coordinate basis for $T\Sigma$. Extend this basis to the variation $\Sigma_t$ by extending $\nu$ to $\nu_t \in T\Sigma_t^\perp$ and $\ell_t = w + \nu_t$, and extend $\partial_i$ so that it remains a coordinate basis for $T\Sigma_t$ by Lie-dragging it along the integral curves of $X$, i.e.,
\begin{equation}
\label{eqC.3}
\pounds_X\partial_i = 0, \quad i = 3, \dotsc, n.
\end{equation}
Therefore $\ell_t$ remains orthogonal to $\partial_i$ and thus
\begin{equation}
\label{eqC.4}
\begin{split}
0=&\, \nabla_X\left (\ell_t\cdot \partial_i\right ) = \pounds_X\left ( \ell_t \cdot \partial_i\right ) = \partial_i \cdot\pounds_X \ell_t + \cancel{\ell_t \cdot \pounds_X\partial_i} + \pounds_Xg \left (\ell_t,\partial_i\right )
\\
=&\, \partial_i \cdot \left ( \nabla_X \ell_t - \nabla_{\ell_t}X\right ) + \partial_i \cdot \nabla_{\ell_t}X + \ell_t \cdot \nabla_{\partial_i}X
\\
=&\, \partial_i \cdot \nabla_X \ell_t + \ell_t \cdot \nabla_{\partial_i}X.
\end{split}
\end{equation}
(In the penultimate equality, $\nabla_{\ell_t}X$ only makes sense when $X$ is defined on a spacetime neighbourhood of $\Sigma$. But the term cancels, so the result does not depend on the choice of smooth extension of $X$ off of $\Sigma$.)

Evaluating \eqref{eqC.4} at $t = 0$ gives
\begin{equation}
\label{eqC.5}
0 = \varphi\partial_i \cdot \nabla_\nu \ell + \ell \cdot \nabla_{\partial_i}(\varphi \nu) \,=\, \varphi\partial_i \cdot \nabla_\nu \ell + \ell \cdot \big ( (\partial_i\varphi) \nu + \varphi \nabla_{\partial_i}\nu \big ) .
\end{equation}
Since $\nu$ is normalized, $\nabla_{\partial_i}\nu \in \text{span}\left \{ w, \partial_3, \dotsc, \partial_n\right \}$. Hence $\ell \cdot \nabla_{\partial_i} \nu = w \cdot \nabla_{\partial_i} \nu$. Also $\ell \cdot \nu =  1$. Therefore
\begin{equation}
\label{eqC.6}
\partial_i \cdot \nabla_\nu\ell = -\partial_i \log\varphi -  w\cdot \nabla_{\partial_i}\nu.
\end{equation}

\begin{lemma}\label{lemmaC.1}
If $\nabla_\ell f = 0$ on $\Sigma$, then
\begin{equation}
\label{eqC.7}
\nabla_{\nu}\ell\cdot \nabla f \big\vert_{\Sigma} = -D f \cdot D\log\varphi +  K(Df, \nu),
\end{equation}
where $D$ denotes the gradient with respect to the metric on $\Sigma$ and $K$ is the second fundamental form of $V$ within $M$.
\end{lemma}

\begin{proof}
From $\eqref{eqC.6}$ we immediately have
\begin{equation}
\label{eqC.8}
\nabla_{\nu}\ell\cdot \nabla f \big\vert_{\Sigma} = \nabla_\Sigma^\perp f \cdot \nabla_\nu \ell -D f \cdot D\log\varphi -  w\cdot \nabla_{Df}\nu,
\end{equation}
where $\nabla^\perp_\Sigma  = \nabla  - D$. The third term on the right equals $ K(Df, \nu)$ by definition. Furthermore,
\begin{equation}
    \label{eqC.9}
    0 = \nabla_\ell f = \nabla f \cdot \ell = \nabla_\Sigma^\perp f \cdot \ell ,
\end{equation}
so $\nabla_\Sigma^\perp f$ is parallel to $\ell$. But $\ell \cdot \nabla_\nu \ell = 0$ since $\ell_t \cdot \ell_t = 0$. Hence the first term on the right of \eqref{eqC.8} vanishes.
\end{proof}

\end{document}